\RequirePackage{fix-cm}
\documentclass[smallextended]{svjour3}       
\smartqed  
\usepackage[text={6.6in,9.6in},centering]{geometry}
 \usepackage{mathptmx}      
\usepackage{graphicx}
\usepackage{amssymb}
\usepackage{amsmath}
\usepackage{natbib}

\usepackage{hyperref}
\hypersetup{
    colorlinks,%
    citecolor=blue,%
    filecolor=black,%
    linkcolor=blue,%
    urlcolor=blue
}
\newcommand{\nbt}{\varphi}

\newcommand{\pp}{\mathbb{P}}

\newcommand{\ind}{\mathbb{I}}
\newcommand{\tsp}{\mathcal{T}}
\newcommand{\utsp}{\mathcal{T}^*}

\newcommand{\pyule}{\mathbb{P}_{\text{YHK}}}
\newcommand{\puni}{\mathbb{P}_{\text{PDA}}}

\newcommand{\puy}{\mathbb{P}_{\text{uYHK}}} 
\newcommand{\puu}{\mathbb{P}_{\text{uPDA}}}
\newcommand{\rev}[1]{#1}
\newcommand{\tw}{}

\usepackage{marvosym}
\newcommand{\envelope}{(\raisebox{-.5pt}{\scalebox{1.45}{\Letter}}\kern-1.7pt)}
 \journalname{Journal of Mathematical Biology}

\begin{document}

\title{Clades and clans: a comparison study of two evolutionary models
\thanks{SZ was supported in part by the New Zealand Marsden Fund,
CT  by the National Science Foundation contract DBI-1146722, and TW by the Singapore MOE grant R-146-000-134-112. 
}
}


\author{Sha Zhu         \and
        Cuong Than \and 
        Taoyang Wu
}


\institute{Sha Zhu \at
Wellcome Trust Centre for Human Genetics, University of Oxford,  United Kingdom\\
              \email{sha.joe.zhu@gmail.com}           
           \and
           Cuong Than \at
             Department of Computer Science, University of Tuebingen, Germany\\
             \email{thvcuong@gmail.com}
            \and
           Taoyang Wu \envelope \at
            School of Computing Sciences, University of East Anglia,  United Kingdom\\
            \email{taoyang.wu@gmail.com}
}

\date{Received: date / Accepted: date}

\maketitle

\begin{abstract}
The Yule-Harding-Kingman (YHK) model and the proportional to distinguishable arrangements (PDA) model are two binary tree generating models that are widely used in evolutionary biology.
Understanding the distributions of clade sizes under these two models 
provides valuable insights into macro-evolutionary processes, and is important in hypothesis testing and Bayesian analyses in phylogenetics. 
Here we show that these distributions are log-convex, which implies that very large clades or very small clades are more likely to occur under these two models.  Moreover, we prove that there exists a critical value $\kappa(n)$ for each $n\geqslant 4$  such that 
for a given clade with size $k$, 
 the probability that this clade is contained  in a random tree with $n$ leaves  generated under the YHK model is higher than that under the PDA model if $1<k<\kappa(n)$, and lower if  $\kappa(n)<k<n$. 
Finally, we extend our results to binary unrooted trees, and obtain similar results for the distributions of clan sizes.

\end{abstract}

\section{Introduction}
\label{intro}
Distributions of genealogical features such as shapes, subtrees, and clades are of interest in phylogenetic and  population genetics. By comparing biological data with these distributions, which can be derived from null models such as the Yule-Harding-Kingman (YHK) model and proportional to distinguishable arrangements (PDA) model, we can obtain insights into  macro-evolutionary processes underlying the data \citep{felsenstein04a,mooers97a,mooers02a,nordborg98a,nordborg01a}. For instance, phylogenetic tree statistics were used to study  variation in speciation and extinction rates (see, e.g.~\citet{agapow02a,mooers97a,rogers96a}). 

As a basic concept in phylogenetic studies and systematic classification of species, a clade, also known as a monophyletic group, is a subset of extant species containing all the descendants of a common ancestor. In this paper, we are interested in the distributions of clade size in a random tree generated under the null models. 
Such distributions have been utilized in hypothesis testing as to whether a set of extant \tw{taxa} forms a clade \citep{hudson02a,rosenberg07a}, and 
are relevant to the Bayesian approach to phylogenetic reconstruction \citep{PR05,PS06}. 

Two well-studied and  commonly used null models in evolutionary biology are the Yule-Harding model~\citep{yule25a, harding71a}  and the PDA model  (also known as the uniform model) \citep{Aldous2001}. Loosely speaking, under the PDA model all rooted binary trees are chosen with equal probabilities, while under the Yule-Harding model each tree is chosen with a probability proportion to the number of total orders that can be assigned to internal nodes of the tree so that the relative (partial) order is preserved~\citep[see, e.g.~][]{semple03a}. 
More precisely, the Yule-Harding model assumes a speciation process  with a constant pure-birth rate \citep{Blum2006,Pinelis2003}, \tw{ which generates the same probability distributions of tree topologies as  Kingman's coalescent process \citep{Kingman1982}.
Therefore, we will refer to it as the Yule-Harding-Kingman (YHK) model~\citep{aldous96a}. }
Both the YHK model and PDA model are used to generate prior probabilities of tree topologies in Bayesian phylogenetic analyses \citep{Li2000,Rannala1996}.

Comparison studies of various tree statistics between the YHK and  PDA models have been reported in the literature. For example,
\citet{McKenzie2000} derive the asymptotic probability distributions of cherries in phylogenetic trees; \citet{Steel2012} discusses the root location in a random Yule or PDA tree; \citet{Blum2006} obtain formulas for the mean, variance, and covariance of the Sackin \citep{Sackin1972} and Colless \citep{Colless1982} indices, two popular indices used to measure the balance of phylogenetic trees.

Note that in Bayesian analyses, the output is often  clade support calculated from the consensus of the approximated posterior distribution of the topologies. However, the relationships between topological priors and clade priors are often not straightforward. For instance, it is observed that the uniform topological prior, which is induced by the PDA model, \tw{leads to}  non-uniform clade priors \citep{PR05}. Indeed, for $n>4$, neither the PDA model nor the YHK model  gives rise to \tw{a  uniform prior} on clades \citep{PS06}. \tw{As an attempt to further elucidate these relationships, in this paper we study the distributions of clade sizes in the PDA model, and then conduct a comparison study of these distributions with those in the YHK model. In addition, we conduct a similar study on clans, the counterpart of clades for unrooted trees.}

The remainder of the paper is organized as follows. Sections 2 and 3 contain necessary notation and background used in the paper and a brief review of the YHK and PDA models. We then present in Section 4 the results concerning clade probabilities under the two null models, 
and those related to clan probabilities in Section 5. Finally, we conclude in Section 6 with discussions and remarks.

\section{Preliminaries}
\label{sec:preliminaries}
In this section, we present some basic notation and background concerning phylogenetic trees and log-convexity that will be used in this paper.
From now on, $X$ will be used to denote the leaf set, and we assume that $X$ is a finite set of size $n=|X|\geqslant 3$ unless stated otherwise. 

\bigskip
\noindent
\subsection{Phylogenetic trees}

A {\em tree} is a connected acyclic graph. A vertex will be  referred to  as a {\em leaf} if its degree is one, and an {\em interior vertex} otherwise. \tw{An unrooted tree is {\it binary} if all interior vertices have degree three. A  {\em rooted} tree is a tree that has exactly one distinguished node designated as the {\em root}, which is usually denoted by $\rho$. A rooted tree is binary if the root has degree two and all other interior vertices have degree three. }

A {\em phylogenetic tree} on $X$ is a binary tree with leaves bijectively labeled by elements of $X$. 
The set of rooted and unrooted phylogenetic trees on $X$ are denoted by $\tsp_X$ and $\utsp_X$, respectively.  Two examples of phylogenetic trees on $X=\{1,\dots,7\}$,  one rooted and the other unrooted, are presented in Figure \ref{fig:trees}.

\begin{figure}
	\centering
	\begin{tabular}{cc}
		\includegraphics[scale=0.7]{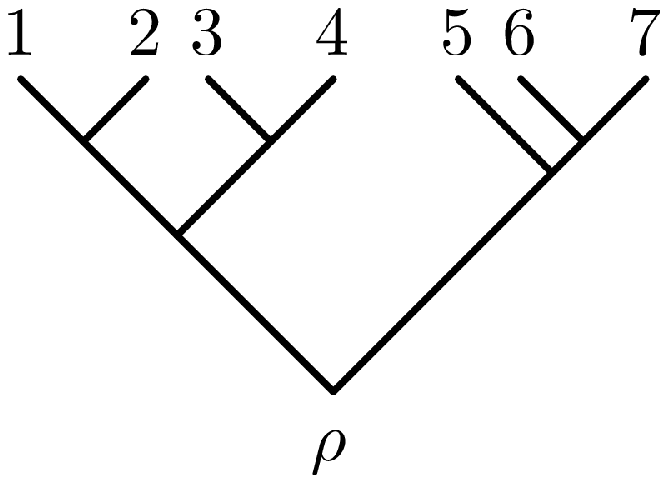} & \includegraphics[scale=0.7]{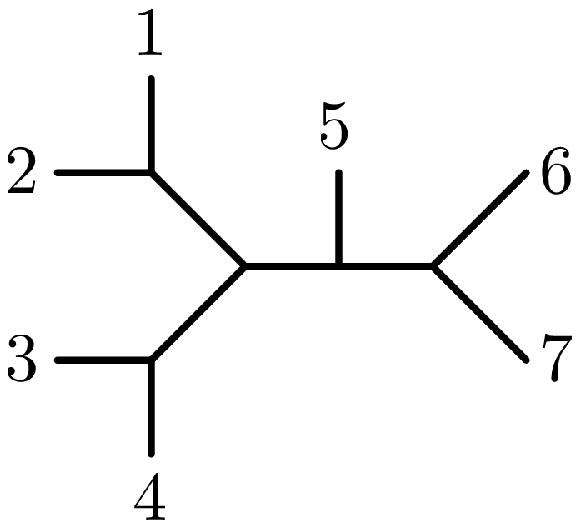}
	\end{tabular}
	\caption{Example of a rooted phylogenetic tree (left) and an unrooted phylogenetic tree (right).  
	}
	\label{fig:trees}
\end{figure}

Let $T$ be a rooted phylogenetic tree on $X$.
Given two vertices $v$ and $u$ in tree $T$,  $u$ is  {\em below} $v$ if $v$ is contained in the path between $u$ and the root of $T$. In this case, we also say $u$ is a {\em descendant} of $v$ if $v$ and $u$ are distinct. 
A {\em clade} of $T$ is a subset of $X$ that contains precisely all the leaves  below a vertex in $T$.   
A clade $A$ is called trivial if $|A|=1$ or $|A|=X$ holds, and non-trivial otherwise. 
Since $T$ has $2n-1$ vertices, it contains precisely $2n-1$ clades, including $n+1$ trivial ones. For example, the rooted phylogenetic tree on $X=\{1,\dots,7\}$ depicted in Figure~\ref{fig:trees} has 13 clades: the five non-trivial ones are $\{1,2\}, \{3,4\}, \{1,2,3,4\}, \{6,7\}$ and $\{5,6,7\}$.

Suppressing the root of a tree $T$ in $\tsp_X$, that is, removing $\rho$ and replacing the two edges incident with $\rho$ with an edge connecting the two vertices adjacent to $\rho$, results in an unrooted tree in $\utsp_X$, which will be denote by $\rho^{-1}(T)$.  For instance, for the rooted tree $T$ and unrooted tree $T^*$ in Figure \ref{fig:trees}, we have $T^*=\rho^{-1}(T)$. Note that for each $T^*$ in $\utsp_X$, there are precisely $2n-3$ rooted trees $T$ in $\tsp_X$ such that $T^*=\rho^{-1}(T)$ holds.

Recall that a {\em split} $A|B$ on $X$ is a bipartition of $X$ into two disjoint non-empty sets $A$ and $B$, that is, $A\cap B=\emptyset$ and $A\cup B=X$. Let $T^*$ be an unrooted tree in $\utsp_X$. Every edge $e$ of $T^*$ induces a necessarily unique split $A|B$ of $X$ obtained as the two sets of leaves separated by $e$. \tw{In other words, the path between a pair of leaves in $X$ contains $e$ if and only if one of these two leaves is in $A$ and the other one is in $B$.}  In this case, we say $A|B$ is a split contained in $T^*$.
A {\em clan} $A$ of $T^*$ is a subset of $X$ such that $A|(X\setminus A)$ is a split contained in $T^*$.  
 Since $T^*$ has $2n-3$ edges and each edge induces two distinct clans, it contains precisely $2(2n-3)$ clans.

\bigskip
\noindent
\subsection{Log-convexity} 
A sequence $\{y_1,\dots,y_m\}$ of real numbers is called {\em positive} if each 
number \tw{contained} in the sequence is \tw{greater than} zero. 
It is called {\em log-convex} if  
 $y_{k-1}y_{k+1}\geqslant y_k^2$ holds for $2\leqslant k \leqslant m-1$.  Clearly, a positive sequence $\{y_k\}_{1\leqslant k \leqslant m}$ is log-convex if and only if the sequence $\{y_{k+1}/y_k\}_{1\leqslant k \leqslant m-1}$ is increasing. Therefore, a log-convex sequence \tw{of positive numbers}  is necessarily {\em unimodal}, that is, there exists an index $1\leqslant k \leqslant m$ such that 
\begin{equation}
\label{def:unimodal}
y_1\geqslant y_2 \geqslant \dots \geqslant y_k~~~\mbox{and}~~~y_k \leqslant y_{k+1} \leqslant \cdots \leqslant y_m
\end{equation}
hold. 
Recall that a sequence $\{y_i\}_{1\leqslant i \leqslant m}$ is also called unimodal if 
$y_1\leqslant y_2 \leqslant \dots \leqslant y_k$ and $y_k \geqslant y_{k+1} \geqslant \cdots \geqslant y_m$ hold for some $1\leqslant k\leqslant m$. However, in this paper,  unimodal is always referred to the situation specified in Eq.~(\ref{def:unimodal}).

For later use, we end this section with the following results concerning log-convex sequences (see, e.g.~\citet{LW}). 

\begin{lemma}
\label{lem:log-convex}
If $\{y_i\}_{1\leqslant i \leqslant m}$ and $\{y'_i\}_{1\leqslant i \leqslant m}$ are two positive and log-convex sequences, then the sequences $\{y_i+y'_i\}_{1\leqslant i \leqslant m}$ and $\{y_i\cdot y'_i\}_{1\leqslant i \leqslant m}$ are positive and log-convex.
\hfill $\square$
\end{lemma}

\section{The PDA and YHK models}

In this section, we present a formal definition of the two null models investigated in this paper:  the {\it proportional to distinguishable arrangements} (PDA) model and {\it Yule--Harding--Kingman} (YHK) model.

To begin with, recall that the number of rooted phylogenetic trees with leaf set $X$ with $n=|X|$ is 
$$\varphi(n):= (2n-3)!! = 1\cdot 3 \dotsb (2n-3)=\frac{(2n-2)!}{2^{n-1}(n-1)!}.$$ 
Here we will use the convention that $\varphi(1)=1$.  Under the PDA model, each tree has the same probability to be generated, that is, we have 
\begin{equation} \label{eq:rooted-pda-prob}
	\puni(T) = \frac{1}{\varphi(n)}
\end{equation}
for every $T$ in $\tsp_X$.

Under the Yule--Harding model, 
a rooted phylogenetic tree on $X$ is generated  as follows.  Beginning with a two leafed tree, we ``grow'' it by repeatedly splitting a leaf into two new leaves.  The splitting leaf is chosen randomly and uniformly among all the present leaves in the current tree.  After obtaining an unlabeled tree with $n$ leaves, we label each of its leaves with a label sampled randomly uniformly (without replacement) from $X$.  When branch lengths are ignored, the Yule--Harding model is shown by~\citet{aldous96a} to be equivalent to the trees generated by Kingman's coalescent process, 
and so we call it the YHK model. Under this model, the probability of generating a tree $T$ in $\tsp_X$ is \citep{semple03a}:
\begin{equation} \label{eq:rooted-yule-prob}
	\pyule(T) = \frac{2^{n-1}}{n!}\prod_{v \in \mathring{V}(T)} \frac{1}{\lambda_v},
\end{equation}
where $\mathring{V}(T)$ is the set of interior nodes of $T$, 
and $\lambda_v$ is the number of interior nodes of $T$ that are below $v$.  
For example, the probability of the rooted tree in Figure~\ref{fig:trees} is
$
	{2^{7-1}}/{(7!\times 3\times 2\times 6)}.
$

\bigskip
\tw{For an unrooted tree $T^*$ in $\utsp_X$, let $\rho(T^*)$ denote the set of rooted trees $T$ in $\tsp_X$ with $T^*=\rho^{-1}(T)$.
As noted previously in Section~\ref{sec:preliminaries}, $T^*$ can be obtained from each of the $2n-3$ rooted trees $T$ in $\rho(T^*)$ by removing the root of $T$. Using this correspondence scheme, a probability measure $\pp$ on  $\tsp_X$ induces a probability \tw{measure} $\pp_u$ on the set $\utsp_X$. 
That is, we have
\begin{equation} \label{eq:unrooted-yule-prob}
	\pp_u(T^*) = \sum_{T\in \rho(T^*)} \pp(T).
\end{equation}
} In particular, \tw{let $\puy$ and $\puu$ denote the probability measures on $\utsp_X$ induced by $\pyule$ and $\puni$, respectively.}
Note that this implies
\begin{equation} \label{eq:unrooted-pda-prob}
	\puu(T^*) = \frac{1}{\varphi(n-1)}
\end{equation}
for every $T^*$ in $\utsp_X$.  Since the number of unrooted phylogenetic trees on $X$ is $|\utsp_X| =\varphi(n-1)= (2n-5)!!$, each tree in $\utsp_X$ has the same probability under $\puu$.

\tw{We end this section with a property of the PDA and YHK models that will play an important role in obtaining our results. Recall that a probability measure $\pp$ on $\tsp_X$ has the {\it exchangeability property} if $\pp$ depends only on tree shapes, that is, 
if two rooted trees $T'$ and $T$ can be obtained from each other by permuting their leaves, then $\pp(T)=\pp(T')$ holds. 
Similarly, a probability measure on $\utsp_X$ has the {exchangeability property} if it depends only on tree shapes. 
It is well-known that  both $\pyule$ and $\puni$, the probability measures on the set of rooted trees $\tsp_X$ induced by the YHK and PDA models, have the exchangeability property~\citep{aldous96a},
By Eqs.~\eqref{eq:unrooted-pda-prob} and \eqref{eq:unrooted-yule-prob}, we can conclude that the probability measures $\puy$ and $\puu$  on the set of unrooted trees $\utsp_X$ also have the exchangeability property.  
}

\section{Clade probabilities}
\label{sec:clade}
In this section, we shall present our main results on clade probabilities. To this end, we need some further notation and definitions.  Given a rooted binary tree $T$, 
let 
\begin{equation}
	\ind_T(A) = \begin{cases} 1, &\text{if $A$ is a clade of $T$},\\ 0, &\text{otherwise,} \end{cases}
\end{equation}
be the `indicator' function that maps a subset $A$ of $X$ to 1 if $A$ is a clade of $T$, and 0 otherwise. Now for a subset $A$ of $X$,  the probability of $X$ being a clade of a random tree sampled according to a probability distribution $\pp$  on $\tsp_X$ is defined as
\begin{equation} \label{eq:clade-prob}
	\pp(A)= \sum_{T \in \tsp_X} \pp(T) \ind_T(A).
\end{equation}
Since $\sum_{A \subseteq X} \ind_T(A) = 2n-1$ for each $T\in \tsp_X$ and 
$\sum_{T \in \tsp_X} \pp(T) = 1$, we have 
\begin{equation*}
	\sum_{A \subseteq X} \pp(A) = \sum_{A \subseteq X} \sum_{T \in \tsp_X} \pp(T) \ind_T(A) = \sum_{T \in \tsp_X} \pp(T) \sum_{A \subseteq X} \ind_T(A)=2n-1.
\end{equation*}
\tw{ By the last equation, we note that each probability measure $\pp$ on $\tsp_X$  induces a measure on the set of all subsets of $X$, which can be normalized to a probability measure by a factor of $1/(2n-1)$. 
}

\bigskip
The above definitions on a subset of $X$ can be extended to a collection 
 of subsets of $X$. That is, given a collection of subsets  $\{A_1, \dotsc, A_k\}$ of $X$, we have 
\begin{equation}
	\ind_T(A_1, \dotsc, A_m) = \ind_T(A_1) \dotsb \ind_T(A_m),
\end{equation}
and
\begin{equation} \label{eq:mulclade-prob}
	\pp(A_1, \dotsc, A_m) =  \sum_{T \in \tsp_X} \pp(T) \big(\ind_T(A_1) \dotsb \ind_T(A_m)\big).
\end{equation}
Note that $\ind_T(A_1, \dotsc, A_m)=1$ if and only if each $A_i$ is a clade of $T$ for $1\leqslant i \leqslant m$. On the other hand, it is well known~(see, e.g.~\citet{semple03a}) that given a collection of subsets  $\{A_1, \dotsc, A_k\}$ of $X$, there exists a tree $T\in \tsp_X$ with $\ind_T(A_1, \dotsc, A_m)=1$  if and only if $\{A_1, \dotsc, A_k\}$ forms a {\em hierarchy}, that is, $A_i\cap A_j\in \{\emptyset, A_i,A_j\}$ holds for $1\leqslant i <j \leqslant m$.

\bigskip
\tw{The following result shows that if a probability measure depends only on tree shapes, then the clade probabilities derived from it are also independent of the `labeling' of the elements.}
\begin{lemma}
\label{lem:set:EP}
Let $\pp$ be a probability measure on $\tsp_X$ that has the \tw{exchangeability property}. Then for each pair of subsets $A$ and $A'$ of $X$ with $|A|=|A'|$, we have
\begin{equation}
\label{eq:set:ep}
\pp(A)=\pp(A')~~~\mbox{and}~~~~~~\pp(A,X\setminus A)=\pp(A',X\setminus A').
\end{equation}
\end{lemma}

\begin{proof}
Suppose that $A$ and $A'$ are two subsets of $X$ that have the same size. Then there exists a permutation $\pi$ on $X$ such that $A'=A^{\pi}:=\{\pi(x)\mid x\in A\}$. Now for each tree $T$ in $\tsp_X$, let $T^\pi$ be the tree obtained from $T$ by relabeling the leaves of $T$ according to permutation $\pi$. Then $A$ is a clade of $T$ if and only if $A^\pi$ is a clade of $T^\pi$. Together with Eq.~(\ref{eq:clade-prob}), we have
\begin{eqnarray*}
\pp(A)&=&\sum_{T\in \tsp_X} \pp(T) \ind_T(A)
=\sum_{T\in \tsp_X} \pp(T) \ind_{T^\pi}(A^\pi)\\
&=&\sum_{T\in \tsp_X} \pp(T^\pi) \ind_{T^\pi}(A^\pi)
=\sum_{T^\pi\in \tsp_X} \pp(T^\pi) \ind_{T^\pi}(A^\pi)=\pp(A^\pi),
\end{eqnarray*}
where the third equality follows from the \tw{exchangeability property} of $\pp$. This shows  
$\pp(A)=\pp(A')$, and a similar argument leads to $\pp(A,X\setminus A)=\pp(A',X\setminus A')$.
\hfill $\square$
\end{proof}

Since $\pyule$ has the \tw{exchangeability property}, by Lemma~\ref{lem:set:EP} we know that $\pyule(A)$ is determined by the size of $A$ only. Therefore,  we denote
\[
	p_n(a) = \pyule(A),
\]
as the probability that a random tree in $\tsp_X$, where $n = |X|$, induces a specific clade $A$ of size $a$ under the YHK model. Similarly, we let 
\[
	q_n(a) = \puni(A),
\]
be the probability that a random tree in $\tsp_X$ induces a specific clade $A$ of size $a$ under the PDA model. 
 In addition, we also denote 
\[
	p_n(a, n-a) = \pyule(A, X\setminus A), \quad \text{and} \quad q_n(a, n-a) = \puni(A, X\setminus A),
\]
the probabilities that both $A$ and $X\setminus A$ are clades of a tree in $\tsp_X$ generated under the YHK and PDA models, respectively.  Note that if both $A$ and $X \setminus A$ are clades of a tree $T$, then they are precisely the clades consisting of the leaves below the two children of the root of $T$.  

\begin{corollary}
\label{cor:set:EP}
Let $\pp$ be a probability measure on $\tsp_X$ that has the \tw{exchangeability property}. For each $1\leqslant a \leqslant n$,  the expected number of clades with size $a$ contained in a random tree sampled according to $\pp$ is 
$${n\choose a} \pp(A),$$
where $A$ is an arbitrary subset of $X$ with $|A|=a$.
\end{corollary}

\begin{proof}
\tw{Denote the collection of subsets of $X$ with size $a$ by $\mathcal{X}_a$ and fix a subset $A\in \mathcal{X}_a$.
Let $Z_T(a):= \sum_{Y\in \mathcal{X}_a} \ind_T(Y)$ be the number of clades with size $a$ contained in a tree $T$. 
  Then the expected number of clades with size $a$ contained in a random tree sampled according to $\pp$ is given by
\begin{eqnarray*}
\sum_{T\in\tsp_X} \pp(T)Z_T(a)=\sum_{T\in \tsp_X}\sum_{Y\in \mathcal{X}_a} \pp(T)\ind_T(Y)=\sum_{Y\in \mathcal{X}_a}\sum_{T\in \tsp_X}\pp(T)\ind_T(Y) 
= \sum_{Y\in \mathcal{X}_a} \pp(Y)={n\choose a} \pp(A),
\end{eqnarray*}
where the last equality holds because by Lemma~\ref{lem:set:EP} we have $\pp(Y)=\pp(A)$ for all $Y\in \mathcal{X}_a$.
\hfill $\square$}
\end{proof}

\subsection{Clade probabilities under the YHK model}
\label{subsec:yule-clade}
In this subsection we study the clade probabilities under the YHK model. First, we have the following theorem concerning the computation of $p_n(a)$ and $p_n(a,n-a)$, which was discovered and rediscovered several times in the literature (see, e.g.,~\cite{blum05a,brown94a, heard92a,rosenberg03a,rosenberg06a}).
\begin{theorem} \label{thm:yule-clade}
For a positive integer $a \leqslant n-1$ we have:
\begin{enumerate}	
	\item[{\rm (i)}] $p_n(a) = \frac{2n}{a(a+1)}\binom{n}{a}^{-1}$.
	\item[{\rm (ii)}] $p_n(a,n-a) = \frac{2}{n-1}\binom{n}{a}^{-1}$.
\end{enumerate}
\end{theorem}

By the above results, we show below that clade probabilities under the YHK model form a log-convex sequence. This implies that the clades with small or large size are more likely to be generated than those with middle size under the model.

\begin{theorem}
\label{thm:yhk:convex}
For $n\geqslant 3$, the sequence $\{p_n(a)\}_{1\leqslant a \leqslant n}$ and
 $\{p_n(a,n-a)\}_{1\leqslant a < n}$ are log-convex. Moreover, let 
 \[
		\Delta(n):=\sqrt{n+\Big(\frac{n-3}{4}\Big)^2 }+\frac{n-3}{4};
	\]
then we have
\begin{enumerate}
	\item[{\rm (i)}] $p_n(a)\geqslant p_n(a+1)$ for $a\leqslant \Delta(n)$, and $p_n(a) < p_n(a+1)$ for $a > \Delta(n)$, and	
	\item[{\rm (ii)}] $p_n(a,n-a)> p_n(a+1,n-a-1)$ for $a \leqslant n/2$ and $p_n(a,n-a)< p_n(a+1,n-a-1)$ for $a \geqslant n/2$.
\end{enumerate}
\end{theorem}

\begin{proof}
Let $y_a=\frac{2n}{a(a+1)}$ for $1\leqslant a \leqslant n-1$ and $y_n=1$, and $y'_a={n \choose a}^{-1}$ for $1\leqslant a \leqslant n$. 
Since $\{y_a\}_{1\leqslant a \leqslant n}$ and  $\{y'_a\}_{1\leqslant a \leqslant n}$ are both log-convex, by Lemma~\ref{lem:log-convex} and Theorem~\ref{thm:yule-clade} we can conclude that the sequence $\{p_n(a)\}_{1\leqslant a \leqslant n}$ is log-convex. 
A similar argument shows that  $\{p_n(a,n-a)\}_{1\leqslant a < n}$ is also log-convex.

By Theorem~\ref{thm:yule-clade}, we have 
\[
\frac{p_n(a+1)}{p_n(a)} = \frac{a(a+1)\binom{n}{a}}{(a+1)(a+2)\binom{n}{a+1}}=\frac{a(a+1)}{(a+2)(n-a)},
\]
for $1 \leqslant a \leqslant n-2$.  The last equation is less than or equal to $1$ if and only if 
$$a(a+1) \leqslant (a+2)(n-a) \iff 2a^2 - (n-3)a - 2n \leqslant 0 .$$
Therefore, $p_n(a+1) \leqslant p_n(a)$ if and only if $a \leqslant \Delta(n)$. This establishes Part (i) of the theorem.

Part (ii) of the theorem follows from the fact that $\binom{n}{a} < \binom{n}{a+1}$ for $a \leqslant n/2 $ and $\binom{n}{a} > \binom{n}{a+1}$ for $a \geqslant  n/2$.
\hfill $\square$
\end{proof}

\subsection{Clade probabilities under the PDA model}
\label{subsec:pda-clade}
Parallel to those in the Section~\ref{subsec:yule-clade}, in this subsection we derive results on clade probabilities under the PDA model. 
\begin{theorem} \label{thm:pda-clade}
For a positive integer $a \leqslant n-1$ we have:
\begin{enumerate}
\item[{\rm (i)}] $q_n(a) = \frac{\varphi(a)\varphi(n-a+1)}{\varphi(n)} = \binom{n-1}{a-1} \binom{2n-2}{2a-2}^{-1}$.
\item[{\rm (ii)}] $q_n(a,n-a)=\frac{\varphi(a)\varphi(n-a)}{\varphi(n)}=\frac{1}{(2n-2a-1)}\binom{n-1}{a-1}\binom{2n-2}{2a-2}^{-1}.$
\end{enumerate}
\end{theorem}

\begin{proof}	 
To derive the formula for $q_n(a)$, it suffices to show that there are $\varphi(a)\varphi(n-a+1)$ trees in $\mathcal{A}$, the subset of trees in $\tsp_X$ containing $A$ as a clade, because the probability of each tree in $\tsp_X$ is $1/\varphi(n)$. Without loss of generality, we can assume that $X=\{1,2,\cdots,n\}$ and $A=\{n-a+1,\cdots,n\}$. Let 
\tw{
$$X':=(X-A)\cup \{n-a+1\}=\{1,2,\cdots,n-a,n-a+1\};
$$
} then each tree in $\mathcal{A}$ can be generated by the following two steps:  picking up a tree in $\tsp_{X'}$ and replacing the leaf with label $n-a+1$ by a tree from $\tsp_{A}$. In addition, a different choice of trees in the first step or the second step will result in a different tree in $\mathcal{A}$. Since there are $\nbt(n-a+1)$ possible choices in the first step and $\nbt(a)$ ones in second step, we can conclude that the number of trees $\mathcal{A}$ is $\nbt(a)\nbt(n-a+1)$.  \tw{ In addition, using the fact that 
$$
\nbt(m)=(2m-3)!!=\frac{(2m-2)!!}{2^{m-1}(m-1)!}
$$
holds for $m\geqslant 1$, we have
$$q_n(a) = \frac{\varphi(a)\varphi(n-a+1)}{\varphi(n)} =
\frac{(2a-2)!(2n-2a)!(n-1)!}{(2n-2)!(a-1)!(n-a)!}
= \binom{n-1}{a-1} \binom{2n-2}{2a-2}^{-1}.$$
}

\tw{	
The proof of the formula for $q_n(a,n-a)$ is similar to the one for $q_n(a)$.  Let $\mathcal{A}^*$ be the collection of the trees in $\tsp_X$ containing both $A$ and $X-A$ as clades. Then a tree in $\mathcal{A}^*$ is uniquely determined by choosing a tree in $\tsp_A$,  and subsequently another tree from $\tsp_{X-A}$.
}  This implies the number of trees in $\mathcal{A}^*$ is $\nbt(a)\nbt(n-a)$. Hence
\begin{align*}
q_n(a,n-a) &= \frac{\nbt(a)\nbt(n-a)}{\nbt(n)} = \frac{1}{(2n-2a-1)} q_n(a) \\ 
		  &= \frac{1}{(2n-2a-1)} \binom{n-1}{a-1}\binom{2n-2}{2a-2}^{-1}. 
\end{align*}
\hfill $\square$
\end{proof}

Recall that in Theorem~\ref{thm:yhk:convex} we show that clade probabilities under the YHK model form a log-convex sequence. Here we establish a similar result for the PDA model, which implies that the sequences $\{q_n(a)\}_{1\leqslant a <n}$ and $\{q_n(a,n-a)\}_{1\leqslant a <n}$ are also unimodal.

\begin{theorem} 
\label{thm:pda:convex}
For $n\geqslant 3$, the sequence $\{q_n(a)\}_{1\leqslant a \leqslant n}$ and
 $\{q_n(a,n-a)\}_{1\leqslant a <n}$ are log-convex. Moreover, 
we have
\begin{enumerate}
\item[{\rm (i)}] $q_n(a+1) \geqslant q_n(a)$ when $a \geqslant n/2$, and $q_n(a+1) \leqslant q_n(a)$ when $a \leqslant n/2$. 
\item[{\rm (ii)}] $q_n(a+1,n-a-1) \geqslant q_n(a,n-a)$ when $a \geqslant (n-1)/2$, and $q_n(a+1,n-a-1) \geqslant q_n(a,n-a)$ when $a \leqslant (n-1)/2$.
\end{enumerate}
\end{theorem}

\begin{proof}
By Theorem~\ref{thm:pda-clade} and $q_n(n)=1$, for $1\leqslant a <n$ we have 
	\begin{align*}
		\frac{q_n(a+1)}{q_n(a)} 
			&= \frac{2a-1}{2n-2a-1},
	\end{align*}
	which is greater than or equal to $1$ when $2a-1 \geqslant 2n-2a-1$, or equivalently when $a \geqslant n/2$. Thus Part (i) follows. Moreover, we have
\begin{align*}
\frac{q_n(a+1)q_n(a-1)}{q^2_n(a)} = \Big(\frac{2a-1}{2a-3}\Big)\Big(\frac{2n-2a+1}{2n-2a-1}\Big)\geqslant 1,
\end{align*}
for $2\leqslant a <n$, and hence $\{q_n(a)\}_{1\leqslant a \leqslant n}$ is log-convex.

	Similarly, we have
	\[
		\frac{q_n(a+1,n-a-1)}{q_n(a,n-a)} = \Big(\frac{2n-2a-1}{2n-2a-3}\Big) \Big(\frac{q_n(a+1)}{q_n(a)} \Big)= \frac{2a-1}{2n-2a-3},
	\]
	which is greater than or equal to $1$ when $2a-1 \geqslant 2n-2a-3$, or equivalently when $a \geqslant (n-1)/2$. Moreover, we have
\begin{align*}
\frac{q_n(a+1,n-a-1)q_n(a-1,n-a+1)}{q^2_n(a,n-a)} =
\Big( \frac{2a-1}{2a-3}\Big)\Big(\frac{2n-2a-1}{2n-2a-3}\Big)\geqslant 1,
\end{align*}
and hence $\{q_n(a)\}_{1\leqslant a < n}$ is log-convex.
\hfill $\square$
\end{proof}

\subsection{A comparison between the PDA and YHK models}
Using the formulae for computing clade probabilities under the PDA and YHK models presented in the previous two subsections,  here we investigate the differences between these two models. Let's begin with comparing  $p_n(a)$ and $q_n(a)$, the probabilities of a specific (and fixed) clade of size $a$ under the YHK and PDA models, respectively. As an example, consider the ratio of $p_n(a)/q_n(a)$ with $n=30$ as depicted in Figure~\ref{fig:pq}. Then it is clear that, except for $a=1$ \tw{for which} both $p_n(a) = q_n(a) = 1$, the ratio is strictly decreasing and is \tw{less than}  $1$ when $a$ is greater than certain value.  This `phase transition' type phenomenon holds for all $n>3$, as the following theorem shows.

\begin{figure}
	\centering
	\includegraphics[scale=0.7]{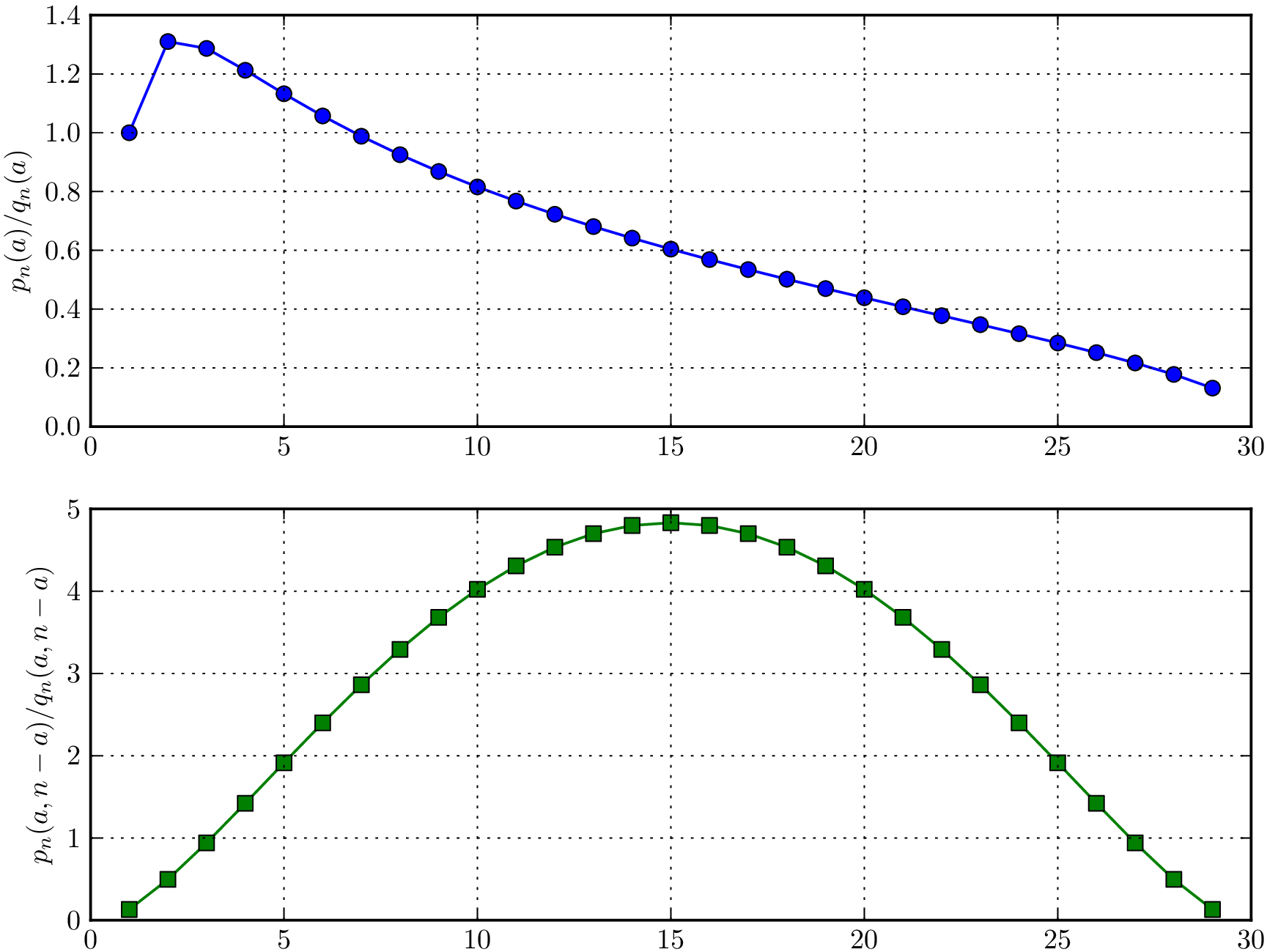}
	\caption{Plots of the ratios $p_n(a) / q_n(a)$ and $p_n(a,n-a)/q_n(a,n-a)$, with $n=30$ and $a=1,\dotsc,29$.}
	\label{fig:pq}
\end{figure}

\begin{theorem}
\label{thm:clade:comp}
For $n> 3$, there exists a number $\kappa(n)$ in $[2,n-1]$, such that $p_n(a)>q_n(a)$ for $2\leqslant a<\kappa(n)$, and $p_n(a)<q_n(a)$ for $\kappa(n)<a \leqslant n-1$.
\end{theorem}

\begin{proof}
\rev{Let 
\[
	g_n(a) = \frac{p_n(a)}{q_n(a)} = \frac{2n}{a(a+1)} \binom{2n-2}{2a-2} \binom{n}{a}^{-1} \binom{n-1}{a-1}^{-1}.
\]
Using the identity $\binom{m}{k+1} = \frac{m-k}{k+1}\binom{m}{k}$, we obtain
\[
	\frac{g_n(a+1)}{g_n(a)} = \frac{a(a+1)(2n-2a-1)}{(a+2)(2a-1)(n-a)}.
\]
We have 
\[
	a(a+1)(2n-2a-1) < (a+2)(2a-1)(n-a) \iff a > \frac{2n}{n+3},
\]
and hence $g_n(a) > g_n(a+1)$ for $2n/(n+3) < a \leqslant n-2$. Since $2n/(n+3) < 2$, we have $g_n(2) > g_n(3) > \dotsb > g_n(n-1)$.}

It is easy to see that for $n > 3$,
\[
	g_n(2) = \frac{2(2n-3)}{3(n-1)} > 1
\]
and 
\[
	g_n(n-1) = \frac{2(2n-3)}{n(n-1)} < 1.
\]
This and the fact that $g_n(a)$ is strictly decreasing on $[2,n-1]$ imply the existence of the number $\kappa(n)$ in the theorem.
\hfill $\square$
\end{proof}

Next, we consider $p_n(a,n-a)$ and $q_n(a,n-a)$. 
Note that by definition, both $p_n(a)$ and $q_n(a,n-a)$ are symmetric about $n/2$, as demonstrated by the plot of the ratio $p_n(a,n-a)/q_n(a,n-a)$ with $n=30$ in~Figure~\ref{fig:pq}. In addition, the figure shows that the ratio is strictly increasing on the interval $[1, \lfloor n/2 \rfloor]$ (and by the symmetry of the ratio, it is strictly decreasing on the interval $[\lceil n/2 \rceil, n-1]$).  This observation is made precise and rigorous in the following theorem. 

\begin{theorem}
For $n > 3$, there exists a number $\lambda(n)$ in $[1,\lfloor n/2 \rfloor]$, such that $p_n(a,n-a)<q_n(a,n-a)$ for $1\leqslant a \leqslant \lambda(n)$, and $p_n(a,n-a)>q_n(a,n-a)$ for $\lambda(n)<a \leqslant \lfloor n/2 \rfloor$.
\end{theorem}

\begin{proof}
\rev{
Let 
\[
	h_n(a) = \frac{p_n(a,n-a)}{q_n(a,n-a)} = \frac{2(2n-2a-1)}{n-1} \binom{2n-2}{2a-2} \binom{n}{a}^{-1} \binom{n-1}{a-1}^{-1}.
\]}
Then 
\[
	\frac{h_n(a+1)}{h_n(a)} = \frac{(a+1)(2n-2a-3)}{(2a-1)(n-a)} > 1,
\]
where the last inequality follows from \tw{the observation} that
$$
(a+1)(2n-2a-3)-(n-a) (2a-1)=3(n-2a-1) > 0
$$
holds for $1\leqslant a \leqslant \lfloor n/2 \rfloor - 1$.  This implies that the function $h_n(a)$ is strictly increasing on the interval $[1, \lfloor n/2 \rfloor]$.

Thus, it now suffices to show that $h_n(1) \leqslant 1$ and $h_n(\lfloor n/2 \rfloor) \geqslant 1$ in order to demonstrate the existence of $\lambda(n)$. We have
\[
	h_n(1) = \frac{p_n(1, n-1)}{q_n(1, n-1)} = \frac{2(2n-3)}{n(n-1)} < 1,
\]
if $n > 3$. Let $k = \lfloor n/2 \rfloor$. If $n$ is even (i.e., $k=n/2$), then for $k \geqslant 2$
\begin{align*}
	h_{2k}(k) &= \frac{2(4k-2k-1)}{(2k-1)} \binom{4k-2}{2k-2} \binom{2k}{k}^{-1} \binom{2k-1}{k-1}^{-1} \\
	 &= \binom{4k-2}{2k-2} \binom{2k-1}{k-1}^{-2} > 1.
\end{align*}
The inequality in the last equation can be seen as follows. Let $A$ and $B$ be two sets, each having $(2k-1)$ elements. The number of subsets of $A \cup B$ that have $k-1$ elements from each of $A$ and $B$ is $\binom{2k-1}{k-1}^2$. On the other hand, the total number of $(2k-2)$-subsets of $A \cup B$ is $\binom{4k-2}{2k-2}$.

If $n$ is odd (i.e., $k=(n-1)/2$), then 
\begin{align*}
	h_{2k+1}(k) 
		&= \frac{2(2k+1)}{2k} \binom{4k}{2k-2}\binom{2k+1}{k}^{-1}\binom{2k}{k-1}^{-1}\\
		&= \frac{2k+1}{k} \binom{4k}{2k-2} \frac{k}{2k+1} \binom{2k}{k-1}^{-2}\\
		&= \binom{4k}{2k-2} \binom{2k}{k-1}^{-2}.
\end{align*}
Using the same argument as in proving $h_{2k}(k) > 1$, we also have $h_{2k+1}(k) \geqslant 1$ for $k \geqslant 1$.
\hfill $\square$
\end{proof}

Let $A$ be a fixed subset of $X$ with size $a$, where $1 \leqslant a \leqslant n-1$. 
In the previous two theorems, we present comparison results for $\pp(A)$ and $\pp(A,X\setminus A)$ under \tw{the YHK and PDA models}. We end this subsection with a comparison study of $\pp(A,X\setminus A)/\pp(A)$, that is,
the probability that  a tree $T \in \tsp_X$ sampled according to probability measure $\pp$ contains both $A$ and $X\setminus A$ as its clades (which means that $A$ and $X\setminus A$ are the clades below the two children of the root of $T$), given that $A$ is a clade of $T$. To this end, let
\[
	u_n(a) = \frac{p_n(a,n-a)}{p_n(a)}-\frac{q_n(a,n-a)}{q_n(a)}=\frac{a(a+1)}{n(n-1)}-\frac{1}{2n-2a-1}
\] 
be the difference between the two conditional probabilities under the two models. \tw{We are interested in the sign changes of $u_n(a)$ as it indicates a `phase transitions'  between these two models.  
For instance, considering the values of $u_n(a)$ for $n=30$ as depicted in~Figure~\ref{fig:u},   then there exists a unique change of sign. 
Indeed, the observation that there exists a unique change of sign of $u_n(a)$ holds for general $n$, as the following theorem shows.
}

\begin{figure}
	\centering
	\includegraphics[scale=0.7]{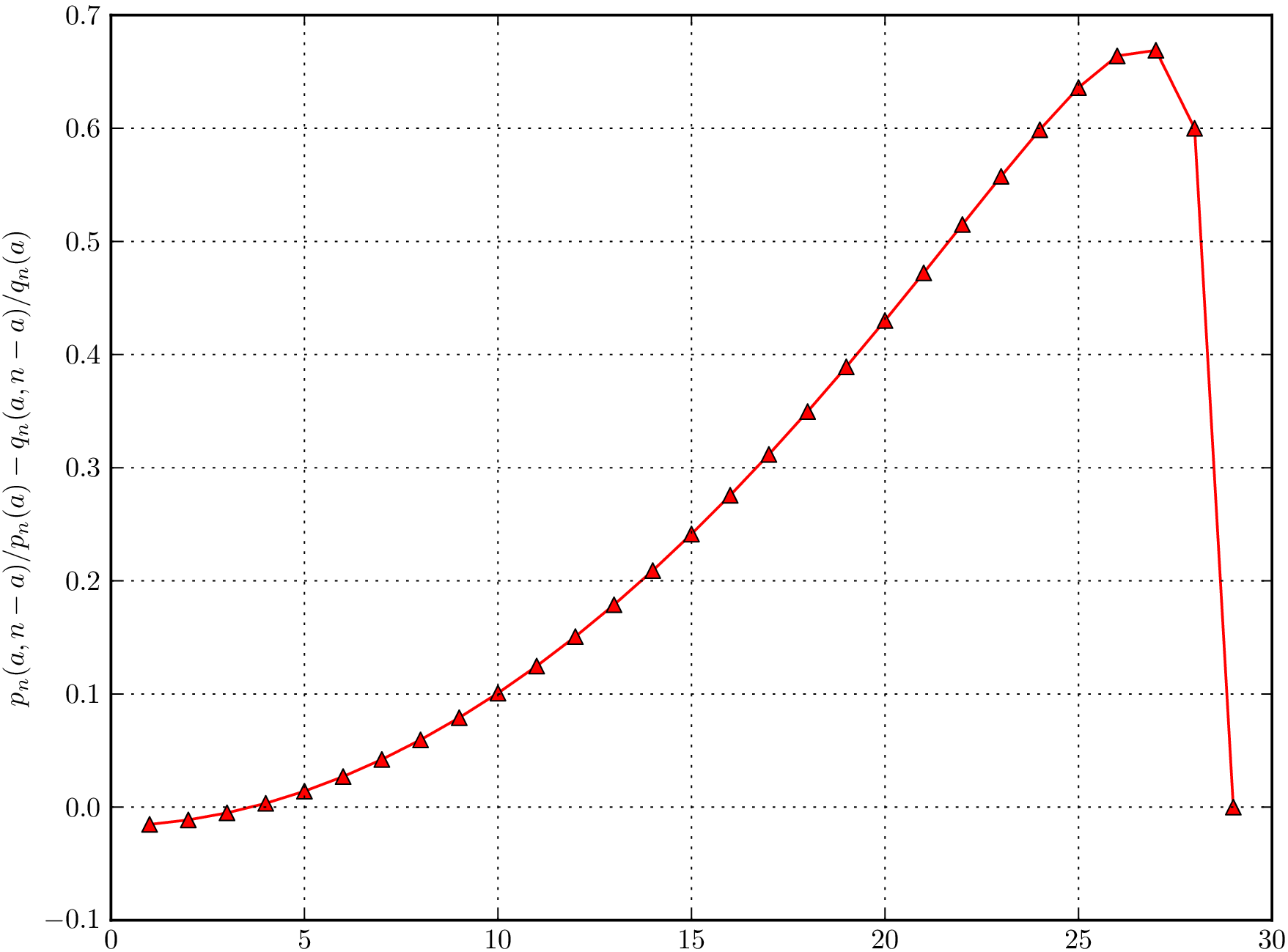}
	\caption{Plot of function $u_n(a)$ 
		with $n=30$. 
		}
	\label{fig:u}
\end{figure}

\begin{theorem}
For $n \geqslant 3$, there exists $\tau(n) \in [1, n-1]$ such that $u_n(a) \leqslant 0$ if $a \leqslant \tau(n)$ and $u_n(a) \geqslant 0$ if $a \geqslant \tau(n)$.
\end{theorem}

\begin{proof}
\rev{Consider the function 
\[
	f_n(x) = \frac{x(x+1)}{n(n-1)} - \frac{1}{2n-2x-1}, \quad x \in \mathbb{R}.
\]
Clearly $f_n(x)$ agrees with $u_n(a)$ when $x = a$. Then
\[
	f_n'(x) = \frac{2x+1}{n(n-1)} - \frac{2}{(2n-2x-1)^2} = \frac{t(2n-t)^2 - 2n(n-1)}{n(n-1)(2n-t)^2},
\]
where $t = 2x+1$. The sign of $f_n'(x)$ thus depends on the sign of
\[
	g_n(t) = t(2n-t)^2 - 2n(n-1).
\]
We see that $g_n(t)$ is a polynomial of $t$ of degree $3$, and hence it can have at most three (real) roots. On the other hand, for $n \geqslant 3$, we have:
\begin{align*}
	g_n(0) &= -2n(n-1) < 0,\\ 
	g_n(1) &= n^2 + (n-1)^2 > 0,\\
	g_n(2n-1) &= -2n(n-2)-1 < 0,
\end{align*}
and 
$$\lim_{t \to \infty} g_n(t) = \infty.$$
Therefore, $g_n(t)$ has exactly three roots $t_1 \in (0, 1)$, $t_2 \in (1, 2n-1)$, and $t_3 > 2n-1$. Note further that $g_n(n) = n^3 - 2n(n-1) = n((n-1)^2 + 1) > 0$, and hence $t_2 > n$. Denoting $x_i = (t_i-1)/2$ for $1\leqslant i \leqslant3$, then we have $f'_n(x)=0$ for $x\in \{x_1,x_2,x_3\}$, 
$f'_n(x)<0$ for $x\in (-\infty,x_1)\cup (x_2,x_3)$, and $f'_n(x)>0$ for $x\in (x_1,x_2)\cup (x_3,\infty)$. } 
\rev{Since $x_1 = (t_1-1)/2 < 0$ and $f_n(a) = u_n(a)$, the sign of $f_n'(x)$ implies that $u_n(1) < u_n(2) < \dotsb < u_n(\lfloor x_2 \rfloor)$. Similarly, we also have $u_n(\lceil x_2 \rceil) > \dotsb > u_n(n-2) > u_n(n-1).$ It is easy to see that for $n \geqslant 3$
\begin{gather*}
	u_n(1) = \frac{2}{n(n-1)} - \frac{1}{2n-3} = -\frac{(n-2)(n-3)}{n(n-1)(2n-3)} \leqslant 0, \\
	u_n(n-1) = \frac{n(n-1)}{n(n-1)} - \frac{1}{2n-2(n-1)-1} = 0.
\end{gather*}
}
\rev{Since $x_2 = (t_2 - 1) / 2 < n-1$ and $x_3 = (t_3 - 1)/2  > n-1$, $\lceil x_2 \rceil \leqslant n-1 < x_3$. This implies that $u_n(\lceil x_2 \rceil) > \dotsb > u_n(n-2) > u_n(n-1) = 0$. Therefore, there exists a positive number $\tau(n) \in [1, x_2]$ such that $u_n(a) \leqslant 0$ if $a \leqslant \tau(n)$ and $u_n(a) \geqslant 0$ if $a \geqslant \tau(n)$.}
\hfill $\square$
\end{proof}

\subsection{Correlation results on the PDA model}
\label{sec:corr:clade:PDA}
In this section, we generalize results in Section~\ref{subsec:pda-clade} for a collection of disjoint subsets of $X$, and then show that the two indicator variables $\ind_T(A)$ and $\ind_T(B)$ are positively correlated.  

\begin{theorem}
\label{thm:partition:pda}
Let $A_1, \dotsc, A_k$ be $k$ disjoint (nonempty) subsets of $X$. 
Denoting $ |A_1|+\dotsb+|A_k|$ by $m$, then we have
\begin{align*}
	\puni(A_1, \dotsc, A_k) 
		&= \frac{\varphi(n-m+k)\prod_{i=1}^k \varphi(\vert A_i \vert)}{\varphi(n)}.
\end{align*}
\end{theorem}

\begin{proof}
We first compute the number of trees that have $A_1, \dotsc, A_k$ as clades. To this end, note that such a tree can be constructed in two steps:
\begin{enumerate}
	\item Build a tree on $\left(X \setminus \bigcup_{i=1}^k A_i \right) \cup\{x_1,\dotsc,x_k\}$, where $x'_1, \dotsc, x'_k$ are leaves not in $X$ serving as ``placeholders'' used in the second step.
	\item Replace each $x'_i$ with a tree in $\mathcal{T}_{A_i}$.
\end{enumerate}
There are $\varphi(n-m+k)$ different choices for a tree in the first step, and  $\prod_{i=1}^k \varphi(|A_i|)$ different ways to replace $x'_1, \dotsc, x'_k$ by trees in $\mathcal{T}_{A_1}, \dotsc, \mathcal{T}_{A_k}$ in the second step. 
Therefore the number of trees that have $A_1, \dotsc, A_k$ as clades is 
$\varphi(n-m+k)\prod_{i=1}^k \varphi(|A_i|)$. 
Together with the fact that each tree in $\mathcal{T}_X$ is chosen with  probability $1/\varphi(n)$ under the PDA model, this implies the theorem.
\hfill $\square$
\end{proof}

Note that $|A_1|+\dotsb+|A_k|=n$ when  $A_1, \dotsc, A_k$ form a partition of $X$. Therefore,  we obtain the following result as a simple consequence of Theorem~\ref{thm:partition:pda} (see {Theorem 5.1} in~\citet{zhu11a} for a parallel result on the YHK model). 
\begin{corollary}
If $A_1, \dotsc, A_k$ form a partition of $X$, then 
\begin{align*}
		\puni(A_1, \dotsc, A_k)  &= \frac{\varphi(k) \prod_{i=1}^k \varphi(|A_i|)}{\varphi(n)}.
	\end{align*}
\end{corollary}

Theorem~\ref{thm:partition:pda} is a general result concerning a collection of clades. When there are only two clades, the below theorem provides a more detailed analysis. 

\begin{theorem}
\label{thm:cor:PDA}
Let $A$ and $B$ be two subsets of $X$ with $a\leqslant b$, where $a=|A|$ and $b=|B|$. 
Then we have

\begin{equation*}
	\puni(A,B) =
	\begin{cases}
	 \frac{\varphi(a)\varphi(n-b+1)\varphi(b-a+1)}{\varphi(n)}, & \text{if $A\subseteq B$,}\\ 
	 \frac{\varphi(a)\varphi(b)\varphi(n-a-b+2)}{\varphi(n)}, & \text{if $A$ and $B$ are disjoint,}\\
	 0, & \text{otherwise.}
	\end{cases}
\end{equation*}
\end{theorem}

\begin{proof}
The first case follows by applying Theorem 2 twice. The second case is a special case of Theorem~\ref{thm:partition:pda}. The third case holds because if $A\cap B\not \in \{A,B, \emptyset\}$, then there exists no tree that contains both $A$ and $B$ as its clades. 
\hfill $\square$
\end{proof}

To establish the last result of this subsection, we need the following technical lemma. 

\begin{lemma}
\label{lem:tree:ineq}
Let $m,n,m',n'$ be positive numbers with $(m-m')(n-n')\geqslant 0$, then
\begin{equation}\label{eq:tree:ineq}
	\nbt(m'+n')\nbt(m+n)\geqslant\nbt(m+n')\nbt(m'+n).
\end{equation}
In particular, if $a\leqslant b\leqslant b'\leqslant a'$ are positive numbers with $a+a'=b+b'$, then we have 
\begin{equation}
	\nbt(a)\nbt(a')\geqslant \nbt(b)\nbt(b').
\end{equation}
\end{lemma}

\begin{proof}
To establish the first claim, we may  assume $m\geqslant m'$ and $n\geqslant n'$, as the proof of the other case, $m \leqslant m'$ and $n\leqslant n'$, is similar. \tw{ Now Eqn.~(\ref{eq:tree:ineq}) holds because we have 
\begin{align}
\frac{\nbt(m+n)}{\nbt(m+n')} 
&=\frac{(2(m+n)-3)\cdot(2(m+n)-5)\cdots 3 \cdot 1}{(2(m+n')-3)\cdot(2(m+n')-5)\cdots 3 \cdot 1} \notag \\
&=(2m+2n-3)(2m+2n-5) \cdots (2m+2n'+1) (2m+2n'-1) \label{eq:ineq:nbt:nprime}\\
& \geqslant (2m'+2n-3)(2m'+2n-5) \cdots (2m'+2n'+1) (2m'+2n'-1) \label{eq:ineq:nbt}\\
&= \frac{\nbt(m'+n)}{\nbt(m'+n')}. \notag
\end{align}
Here Eq.~(\ref{eq:ineq:nbt:nprime}) follows from $n\geqslant n'$ and Eq.~(\ref{eq:ineq:nbt}) from
 $m\geqslant m'$.
 }

The second assertion follows from the first one by \tw{setting} $m'=n'=a/2$, $m=b-a/2$ and $n=b'-a/2$.
\hfill $\square$
\end{proof}

We end this section with the following result, which says that  the random variables $\ind_T(A)$ and $\ind_T(B)$  
are positively correlated when $A$ and $B$ are compatible, that is, $A\cap B\in \{\emptyset, A,B\}$.

\begin{theorem} \label{thm:positive-correlation}
Let $A$ and $B$ be two compatible non-empty subsets of $X$; then 
$$
\puni(A,B)\geqslant \puni(A)\puni(B).
$$ 
\end{theorem}

\begin{proof}
\tw{
Set $a=|A|$ and $b=|B|$. By symmetry we may assume without loss of generality that $a\leqslant b$ holds. Since $A$ and $B$ are compatible, we have either $A\cap B=\emptyset$ or $A\subseteq B$. 
}

\tw{
Since $n-a-b+2\leqslant n-b+1\leqslant n-a+1 \leqslant n$, by Lemma~\ref{lem:tree:ineq} we have
 $$
 \nbt(n)\nbt(n-a-b+2)\geqslant \nbt(n-b+1)\nbt(n-a+1),
 $$
 and hence
\[
	\frac{\nbt(a)\nbt(b)\nbt(n-a-b+2)}{\nbt(n)}\geqslant 	\frac{\nbt(b)\nbt(n-b+1)}{\nbt(n)}\frac{\nbt(a)\nbt(n-a+1)}{\nbt(n)}.
\]
Together with Theorem~\ref{thm:cor:PDA}, this shows that the theorem holds for the case $A\cap B=\emptyset$.}

\tw{
On the other hand, noting that $b-a+1\leqslant b \leqslant n$ and $b-a+1\leqslant n-a+1 \leqslant n$ holds, by Lemma~\ref{lem:tree:ineq} we have
$$	\nbt(n)\nbt(b-a+1)\geqslant\nbt(b)\nbt(n-a+1), 
$$
and hence
\[
	\frac{\nbt(a)\nbt(b-a+1)}{\nbt(b)}\frac{\nbt(b)\nbt(n-b+1)}{\nbt(n)} \geqslant 		\frac{\nbt(b)\nbt(n-b+1)}{\nbt(n)}\frac{\nbt(a)\nbt(n-a+1)}{\nbt(n)}.
\]
Together with Theorem~\ref{thm:cor:PDA}, this shows that the theorem holds for the case $A\subseteq B$, as required.}
\hfill $\square$
\end{proof}

\section{Clan probabilities}
In this section, we study clan probabilities, the \tw{counterpart of clade probabilities for unrooted trees}. To this end, given a subset $A\subseteq X$ and an unrooted tree $T^*\in \utsp_X$, let $\ind_{T^*}(A)$ be the indicator function defined as
 \[
	\ind_{T^*}(A) = \begin{cases} 1, &\text{if $A$ is a clan of $T^*$,}\\ 
						0, &\text{otherwise.} \end{cases}
\]
Then the probability that clan $A$ is contained in a random unrooted tree sampled according to $\pp_u$ is 
$$
\pp_u(A)=\sum_{T^* \in \utsp_X} \pp_u(T^*) \ind_{T^*}(A).
$$
Note that the the clan probability defined as above can be extended to a collection of subsets in a natural way, that is, we have
\begin{equation*}
	\pp_u(A_1, \dotsc, A_m) =  \sum_{T^* \in \utsp_X} \pp_u(T^*) \big(\ind_{T^*}(A_1) \dotsb \ind_{T^*}(A_m)\big).
\end{equation*}

As a generalization of Lemma~6.1 in~\citet{zhu11a}, the following technical result relates clan probabilities to clade probabilities. 
\begin{lemma}
\label{lem:clan:clade}
Suppose that $\pp$ is \tw{a} probability measure on $\tsp_X$ and $\pp_u$ is the probability measure on $\utsp_X$ induced by $\pp$. Then for a nonempty subset $A\subset X$, we have
\[
\pp_u(A)=\pp(A)+\pp(X \setminus A)-\pp(A,X  \setminus A).
\]
\end{lemma}

\begin{proof}
\tw{It is well-known (see, e.g., Lemma~6.1 in~\citet{zhu11a}) that for a rooted binary tree $T$, a set $A$ is a clan of $\rho^{-1}(T)$ if and only if either $A$ is a clade of $T$ or $X  \setminus A$ is a clade of $T$. Now the lemma follows from the definitions and the inclusion-exclusion principle.}
\hfill $\square$
\end{proof}

\bigskip
Now we proceed to studying the clan probabilities under the YHK and PDA models. To begin with, recall that the probabilities of an unrooted tree $T^* \in \utsp_X$ under the YHK and PDA models are
\begin{align*}
	\puy(T^*) = \sum_{T \in \rho(T^*)} \pyule(T)~~\text{and}~~
	\puu(T^*) = \sum_{T \in \rho(T^*)} \puni(T), 
\end{align*}
where $\rho(T^*)$ denotes the set of rooted trees $T$ in $\tsp_X$ with $T^*=\rho^{-1}(T)$.  

By the definition of clan probabilities, we have
\begin{align*}
	\puy(A) &= \sum_{T^* \in \utsp_X} \puy(T^*) \ind_{T^*}(A),~~\text{and}\\
	\puu(A) &= \sum_{T^* \in \utsp_X} \puu(T^*) \ind_{T^*}(A). \notag
\end{align*}
It can be verified, as with the case of clade probabilities, that the \tw{exchangeability property} of $\puy$ and $\puu$ implies that both $\puy(A)$ and $\puu(A)$ depend only on the size $a= |A|$, not on the particular elements in $A$.  Therefore, we will denote them as $p_n^*(a)$ and $q_n^*(a)$, respectively.  

By Lemma~\ref{lem:clan:clade}, we can derive the following formulae to calculate clan probabilities under the two models, \tw{the first of which is established in~\citet{zhu11a}. Note that the second formula reveals an interesting relationship between clan probability and clade probability under the PDA model. Intuitively, it is related to the observation that 
there exists a bijective mapping from $\tsp_X$ to $\utsp_{Y}$ with $Y=X\cup \{y\}$ for some $y \not \in X$ that maps each rooted tree $T$ in $\tsp_X$ 
to the unique tree in $\utsp_Y$ obtained from $T$ by adding the leaf $y$ to the root of $T$.
}

\begin{theorem}
\label{thm:prob:clan}
For $1\leqslant a <n$, we have
\begin{align}
 p^*_n(a) &= 2n\Big[  \frac{1}{a(a+1)}+\frac{1}{(n-a)(n-a+1)}-\frac{1}{(n-1)n} \Big] {n \choose a}^{-1}; \label{eq:clan:yule}\\
q^*_n(a) &= \frac{\nbt(a)\nbt(n-a+1)+\nbt(n-a)\nbt(a+1)-\nbt(a)\nbt(n-a)}{\nbt(n)}\\
	&=\frac{\nbt(a)\nbt(n-a)}{\nbt(n-1)} = \rev{q_{n-1}(a)}\notag. \label{eq:clan:pda}
\end{align}
\end{theorem}
 
 \begin{proof}
Since the first equation is established in~\citet{zhu11a}, it remains to show the second one.  The first equality follows from Lemma~\ref{lem:clan:clade} and Theorem~\ref{thm:pda-clade}. To establish the second equality, it suffices to see that
\begin{align*}
  \nbt(n-1) &[\nbt(a)\nbt(n-a+1)+\nbt(n-a)\nbt(a+1)]\\
  &=\nbt(n-1)\nbt(a)\nbt(n-a) [(2n-2a-1)+(2a-1)]\\ 
  &=\nbt(n-1)(2n-2)\nbt(a)\nbt(n-a)\\
&=(\nbt(n)+\nbt(n-1))\nbt(a)\nbt(n-a). 
\end{align*}
\hfill $\square$
\end{proof}

Recall that in Theorem~\ref{thm:yhk:convex} and~\ref{thm:pda:convex} we show that  the sequence $\{p_n(a)\}_{1\leqslant a < n}$ and
 $\{q_n(a)\}_{1\leqslant a < n}$ are log-convex. The theorem below establishes a similar result for clan probabilities. 

\begin{theorem}
\label{thm:clan:convex}
For $n\geqslant 3$, the sequence $\{p^*_n(a)\}_{1\leqslant a < n}$ and
 $\{q^*_n(a)\}_{1\leqslant a < n}$ are log-convex. Moreover, 
we have 
\begin{enumerate}
	\item[{\rm (i)}] $p^*_n(a)=p^*_n(n-a)$ and $q^*_n(a)=q^*_n(n-a)$ for $1\leqslant a < n$.
	\item[{\rm (ii)}]  $q^*_n(a+1) \leqslant q^*_n(a)$ when $a \geqslant \lfloor (n-1)/2 \rfloor - 1$, and $q^*_n(a+1) \geqslant q^*_n(a)$ when $a \leqslant \lceil (n-1)/2 \rceil$.
\end{enumerate}
\end{theorem}

\begin{proof}
Part (i) follows from Theorem~\ref{thm:prob:clan}. Since $q_n^*(a) = q_{n-1}(a)$ by Theorem~\ref{thm:prob:clan},  Part (ii) and that  $\{q^*_n(a)\}_{1\leqslant a <n}$ is log-convex follow from Theorem~\ref{thm:pda:convex}. 

It remains to show that $\{p^*_n(a)\}_{1\leqslant a <n}$ is log-convex.
To this end, fix a number $n\geqslant 3$, and let $y_a=\frac{1}{a(a+1)}$ for $1\leqslant a <n$. Then clearly $\{y_a\}_{1\leqslant a <n}$ is log-convex. This implies $\{y'_a\}_{1\leqslant a <n}$ with $y'_a=y_{n-a}$ is also log-convex. In addition, since $2y_a\geqslant y_{a+1}+y_{a-1}$ for $2\leqslant a \leqslant n-2$, $\{y^*_a\}_{1\leqslant a <n}$ with $y^*_a=y_a-\frac{1}{n(n-1)}$ is log-convex as well.
By Lemma~\ref{lem:log-convex}, we know $\{y'_a+y^*_a\}_{1\leqslant a <n}$ is log-convex. As $\{{n \choose a}^{-1}\}_{1\leqslant a <n}$ is log-convex, by  Lemma~\ref{lem:log-convex} and Theorem~\ref{thm:prob:clan} we conclude that $\{p^*_n(a)\}_{1\leqslant a <n}$ is log-convex, as required.
\hfill $\square$
\end{proof}

\begin{figure}
	\centering
	\includegraphics[scale=0.7]{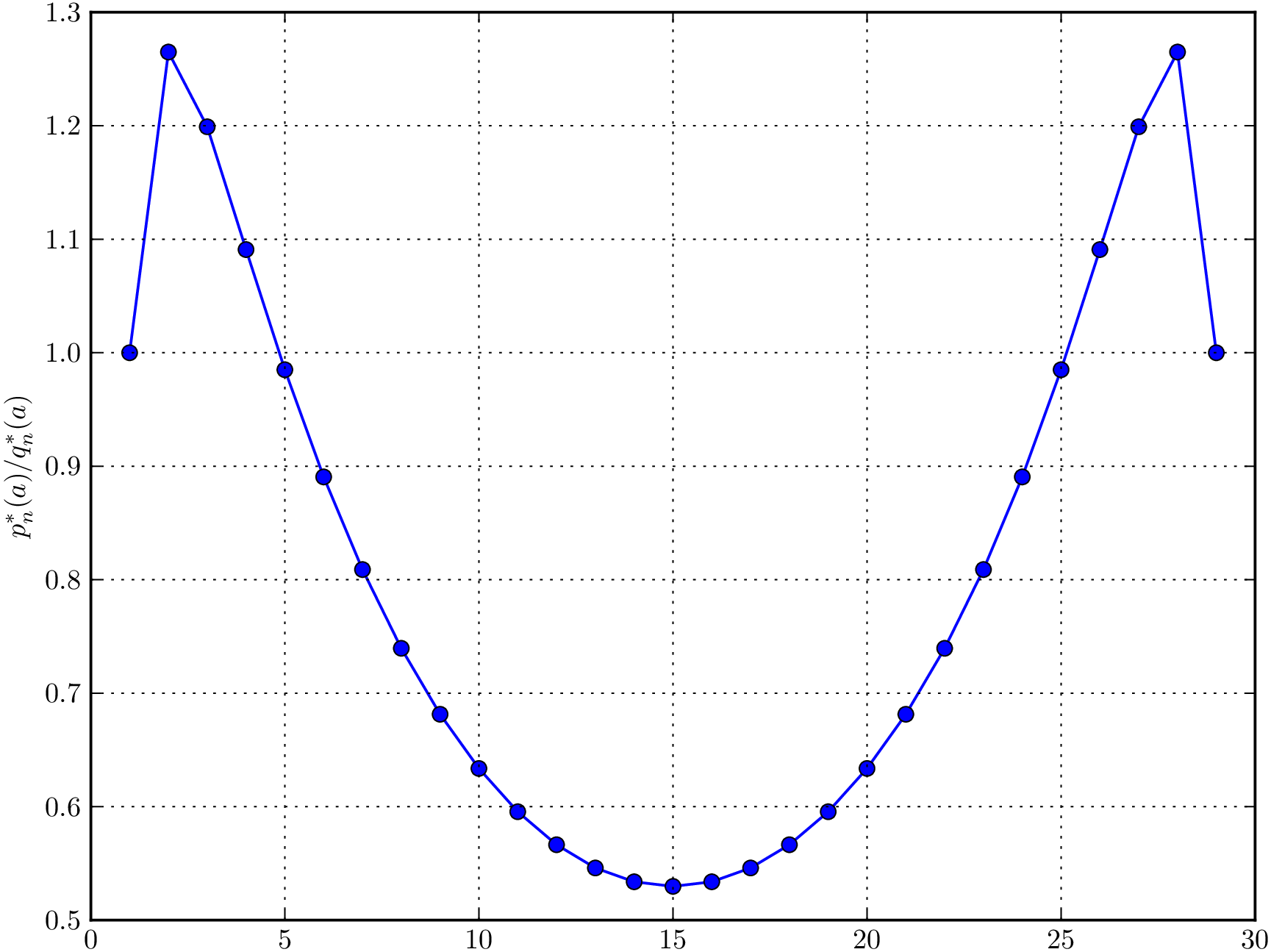}
	\caption{Plot of the ratio $p^*_n(a) / q^*_n(a)$ with $n=30$ and $a=1,\dotsc,29$.}
	\label{fig:pq:star}
\end{figure}

Next, we consider the relationships between clan probabilities under the two models. For instance, consider the ratio of $p^*_n(a)/q^*_n(a)$ with $n=30$ (see Figure~\ref{fig:pq:star}. Then the ratios are symmetric  about $a=15$, which is consistent with Part(i) in Theorem~\ref{thm:clan:convex}. In addition, by the figure it is clear that, \tw{except for $a=1$ for which  $p^*_n(a) = q^*_n(a) = 1$}, the ratio is strictly decreasing on $[2,\lfloor n/2 \rfloor]$ and is \tw{less than} $1$ when $a$ is greater than a critical value. 
We shall show this observation holds for general $n$. To this end, we need the following technical lemma.

\begin{lemma}
\label{lem:clan:comp:bound}
For $n> 5$, we have $p^*_n(\lfloor n/2 \rfloor) < q^*_n(\lfloor n/2 \rfloor)$.
\end{lemma}

\begin{proof}
For simplicity, let $k = \lfloor n/2 \rfloor$.
To establish the lemma, we consider the following two cases.

The first case is when $n$ is even, that is, $n=2k$. Then we have
\begin{align*}
 p_{2k}^*(k) &=4k\Big(\frac{2}{k(k+1)}-\frac{1}{2k(2k-1)}\Big) {2k \choose k}^{-1}\\
 &=\Big( \frac{8}{k+1}-\frac{2}{2k-1} \Big)  {2k \choose k}^{-1}
 =\frac{2(7k-5)}{(k+1)(2k-1)} {2k \choose k}^{-1},
\end{align*}
and 
\begin{align*}
\alpha(k):=\frac{q_{2k}^*(k)}{p_{2k}^*(k)}&=\frac{\nbt(k)\nbt(k)}{\nbt(2k-1)} {2k\choose k} \frac{(k+1)(2k-1)}{2(7k-5)} \\
 &= \frac{(2k-2)!(2k-2)!(2k-2)!(2k)!}{(4k-4)!(k-1)!(k-1)!k!k!}\frac{(k+1)(2k-1)}{2(7k-5)}.
\end{align*}
 Note that  $\alpha(3)=\frac{15}{14}>1$, and $\alpha(k)$ is increasing for $k\geqslant 3$, because 
 \begin{eqnarray*}
 \frac{\alpha(k+1)}{\alpha(k)} &=& \frac{2(2k-1)(2k+1)^2(k+2)(7k-5)}{(4k-1)(4k-3)(k+1)^2(7k+2)} \\
& =& \frac{112k^5+200k^4-116k^3-130k^2+22k+20}
{112k^5+144k^4-59k^3-96k^2+k+6}  \\
&>& 1,
 \end{eqnarray*}
  holds for $k\geqslant 3$.   In other words, for $k\geqslant 3$, we have $\alpha(k)>1$ and hence also ${q_{2k}^*(k)}>{p_{2k}^*(k)}$.

The second case is when $n$ is odd, that is, $n=2k+1$. Then we have
  \begin{eqnarray*}
 p_{2k+1}^*(k)&=&(4k+2)\Big(\frac{1}{k(k+1)}+\frac{1}{(k+1)(k+2)}-\frac{1}{2k(2k+1)}\Big) {2k+1 \choose k}^{-1}\\
& =& \frac{7k+2}{k(k+2)}  {2k+1 \choose k}^{-1},
 \end{eqnarray*}
  and 
 \begin{eqnarray*}
\beta(k):=\frac{q_{2k+1}^*(k)}{p_{2k+1}^*(k)}&=& \frac{\nbt(k)\nbt(k+1)}{\nbt(2k)} {2k+1\choose k}\frac{k(k+2)}{7k+2} \\
 &=& \frac{(2k-2)!(2k-1)!(2k)!(2k+1)!(k+2)}{(4k-2)!(k-1)!(k-1)!k!(k+1)!(7k+2)}.
 \end{eqnarray*}
  Now we have $\beta(3)=25/23>1$. In addition, $\beta(k)$ is increasing for $k\geqslant 3$ by noting that
  \begin{eqnarray*}
\frac{ \beta(k+1)}{\beta(k)} &=& 
\frac{(2k-1)(2k+1)(2k+2)(2k+3)(k+3)(7k+2)}{(k+2)^2(4k+1)(4k-1)k(7k+9)} \\
&=& \frac{112k^6+648k^5+1156k^4+630k^3-152k^2-198k-36}
{112k^6+592k^5+1017k^4+539k^3-64k^2-36k}\\
&\geqslant& 1
  \end{eqnarray*}
  holds for $k\geqslant 3$.   In other words, for $k\geqslant 3$ and $n$ being odd, we also have $\beta(k)>1$ and hence also ${q_{2k+1}^*(k)}>{p_{2k+1}^*(k)}$. This completes the proof.
\hfill $\square$
\end{proof}

Parallel to Theorem~\ref{thm:clade:comp} which compares $p_n(a)$ and $q_n(a)$, the following theorem provides a comparison between $p^*_n(a)$ and $q^*_n(a)$.

\begin{theorem}
\label{thm:clan:comp}
For $n> 5$, there exists a number $\kappa^*(n)$ in $(1,\lfloor n/2 \rfloor)$, such that $p^*_n(a)>q^*_n(a)$ for $2\leqslant a \leqslant \kappa^*(n)$, and $p^*_n(a)<q^*_n(a)$ for $\kappa^*(n)<a \leqslant \lfloor n/2 \rfloor$.
\end{theorem}

\begin{proof}
For simplicity, let $b:=n-a$.
Since we have
 $$
p_n^*(2)=\frac{4\Big(\frac{1}{6}+\frac{2}{n(n-1)(n-2)}\Big)}{n-1}>\frac{2}{3(n-1)}\geqslant \frac{1}{2n-5}=q_n^*(2),
 $$
and $p^*_n(\lfloor n/2 \rfloor) < q^*_n(\lfloor n/2 \rfloor)$ by Lemma~\ref{lem:clan:comp:bound}, 
it suffices to prove that
\[
	g_n(a) = \frac{p^*_n(a)}{q^*_n(a)}
\]
is strictly decreasing on $[2, \lfloor n/2 \rfloor]$. To this end, let 
\[
	f_n(a) = \frac{1}{a(a+1)} + \frac{1}{b(b+1)} - \frac{1}{n(n-1)}.
\]
 From the definition of $g_n(a)$ and Theorem~\ref{thm:prob:clan}, we have
\[
	\frac{g_{n}(a+1)}{g_n(a)} = \frac{f_{n}(a+1)}{f_n(a)} \frac{(a+1)(2b-3)}{b(2a-1)},
\]
which is \tw{less than} $1$ for $2 \leqslant a \leqslant \lfloor n/2 \rfloor - 1$ if and only if
\begin{equation}
\label{eq:clan:comp:pf}
\beta_n(a):= f_n(a) b(2a-1)-f_n(a+1)(a+1)(2b-3) >0~~~~\text{for $2 \leqslant a \leqslant \lfloor n/2 \rfloor - 1$}.
\end{equation}

\bigskip
In the rest of the proof, we shall establish Eq.~(\ref{eq:clan:comp:pf}). To begin with, note that
\begin{align}
\label{eq:Delta}
 \begin{split}
\beta_n(a)=\frac{3}{n-1}-\frac{3(2a+1)}{n(n-1)}+\frac{2a^2+an+5a-2n}{a(a+1)(a+2)}\\
+\frac{2a-3n}{(b-1)(b+1)}+\frac{a+2n+3}{(b-1)b(b+1)}.
 \end{split}
\end{align}

This implies
\begin{align*}
\beta_n(2)
&=\frac{3n^4 - 18n^3 - 39n^2 + 342n - 360}{4n(n-1)(n-2)(n-3)}\\
&=\frac{3n^2(n^2-6n-13)+(342n-360)}{4n(n-1)(n-2)(n-3)}
> 0
\end{align*}
for $n\geqslant 6$ because $\beta_6(2)=1/5$, $\beta_7(2)=24/70$ and $n^2-6n-13>0$ for $n\geqslant 8$.
In addition,  we have
\begin{eqnarray*}
\beta_{2t+1}(t) &=& \frac{4t^2+2t-2}{t(t+1)(t+2)}+\frac{-4t+2}{t(t+2)}+\frac{5}{t(t+2)}>0
\end{eqnarray*}
for $t\geqslant 3$ and
\begin{align*}
\beta_{2t+2}(t) &= 
\frac{3}{2t-1}-\frac{3}{2t+2}+\frac{4t^2+3t-4}{t(t+1)(t+2)}-\frac{4t+6}{(t+1)(t+3)}+\frac{(5t+7)}{(t+1)(t+2)(t+3)} \\
&=\frac{9}{(2t-1)(2t+2)}+\frac{6t^2-12}{t(t+1)(t+2)(t+3)}\\
&> 0
\end{align*}
for $t\geqslant 2$. Therefore, we have $\beta_n(\lfloor n/2 \rfloor - 1)\geqslant 0$ for $n\geqslant 6$. 

It remains to show that $\beta_n(a)$ is strictly decreasing, that is, $\beta_n(a)-\beta_n(a+1)> 0$ for $3\leqslant a \leqslant \lfloor n/2 \rfloor-1$. Indeed, by Eqn.~(\ref{eq:Delta}) we have
\begin{eqnarray*}
\beta_n(a)-\beta_n(a+1) &=& \frac{6}{n(n-1)}+\frac{2a^2+2an+8a-6n}{a(a+1)(a+2)(a+3)}
+\frac{2a^2-6an+4n^2-10n-8}{(b-2)(b-1)b(b+1)} \\
&>&  \frac{n^2-7n-8+2a^2}{(b-2)(b-1)b(b+1)}  \\
&>& 0.
\end{eqnarray*}
Here the first inequality follows from $a\geqslant 3$ and $a \leqslant  \lfloor n/2 \rfloor-1 \leqslant (n-1)/2$ implying $3n^2-6an\geqslant 3n$, and the second one from $a\geqslant 3$ and $n\geqslant 6$. This completes the proof.
\hfill $\square$
 \end{proof}
 
We end this section with some correlation results about clan probabilities under the PDA model.

 \begin{theorem}
\label{thm:partition:unroot:pda}
\tw{Let $A_1, \dotsc, A_k$ be $k$ disjoint (nonempty) subsets of $X$, and let $m = |A_1|+\dotsb+|A_k|$. Then we have}
\begin{align*}
	\puu(A_1, \dotsc, A_k) 
		&= \frac{\varphi(n-m+k-1)\prod_{i=1}^k \varphi(|A_i|)}{\varphi(n-1)}.
\end{align*}
\end{theorem}

\begin{proof}
\rev{Since $\puu(T^*)=1/\varphi(n-1)$ for each tree $T^*$ in $\mathcal{T}_X$, it remains to compute the number of trees that have $A_1, \dotsc, A_k$ as clans is $\varphi(n-m+k-1)\prod_{i=1}^k \varphi(|A_i|)$. To this end, note that such a tree can be constructed in two steps:
\begin{enumerate}
	\item Build an unrooted tree on $\left(X \setminus \bigcup_{i=1}^k A_i \right) \cup\{x_1,\dotsc,x_k\}$, where $x_1, \dotsc, x_k$ are leaves not in $X$ serving as ``placeholders'' used in the second step.
	\item Replace each $x_i$ with a tree in $\mathcal{T}_{A_i}$.
\end{enumerate}
There are $\varphi(n-m+k-1)$ different choices for a tree in the first step, and there are $\prod_{i=1}^k \varphi(a_i)$ different ways to replace $x_1, \dotsc, x_k$ by trees in $\mathcal{T}_{A_1}, \dotsc, \mathcal{T}_{A_k}$. The claim then follows.}
\hfill $\square$
\end{proof}

 \begin{theorem}
\label{thm:cor:unroot:PDA}
Let $A$ and $B$ be two subsets of $X$ with $a\leqslant b$, where $a=|A|$ and $b=|B|$. 
Then we have
\tw{
\begin{equation*}
	\puu(A,B) =
	\begin{cases}
	 \frac{\varphi(b)\varphi(n-b)\varphi(a)\varphi(b-a)}{\varphi(n-1)\varphi(b-1)}, & \text{if $A\subseteq B$,}\\ 
	 \frac{\varphi(a)\varphi(b)\varphi(n-a-b+1)}{\varphi(n-1)}, & \text{if $A$ and $B$ are disjoint,}\\
	 0, & \text{otherwise.}
	\end{cases}
\end{equation*}
}
\end{theorem}

\begin{proof}
The first case follows by applying Theorem~\ref{thm:prob:clan} twice; the second case follows from Theorem~\ref{thm:partition:unroot:pda}.
\hfill $\square$
\end{proof}

\begin{corollary}
Let $A$ and $B$ be two compatible subsets of $X$. Then we have
$$
\puu(A,B)\geqslant \puu(A)\puu(B).
$$ 
\end{corollary}

\begin{proof}
\tw{
Set $a=|A|$ and $b=|B|$. By symmetry we may assume without loss of generality that $a\leqslant b$ holds. Since $A$ and $B$ are compatible, we have either $A\cap B=\emptyset$ or $A\subseteq B$. 
}

\tw{To establish the theorem for the first case, note first that $n-a-b+1\leqslant n-b \leqslant n-a \leqslant n-1$ holds. Therefore by Lemma~\ref{lem:tree:ineq}, we have
$$
\nbt(n-a-b+1)\nbt(n-1)\geqslant \nbt(n-a)\nbt(n-a),
$$
 and hence
 $$
\frac{\nbt(a)\nbt(b)\nbt(n-a-b+1)}{\nbt(n-1)}\geqslant 
\Big(\frac{\nbt(b)\nbt(n-b)}{\nbt(n-1)}\Big)\Big(\frac{\nbt(a)\nbt(n-a)}{\nbt(n-1)}\Big).
$$
Together with Theorem~\ref{thm:cor:unroot:PDA}, this shows that the theorem holds for the case $A\cap B=\emptyset$.
}

\tw{
For the second case, note that $b-a\leqslant n-a \leqslant n-1$ and $b-a \leqslant b-1 \leqslant n-1$ hold. Therefore by by Lemma~\ref{lem:tree:ineq}, we have
$$
\nbt(n-1)\nbt(b-a)\geqslant\nbt(b-1)\nbt(n-a).
$$
and hence
$$
\frac{\varphi(b)\varphi(n-b)\varphi(a)\varphi(b-a)\varphi(n-b)}{\varphi(n-1)\varphi(b-1)} 
\geqslant 
\Big(\frac{\nbt(b)\nbt(n-b)}{\nbt(n-1)}\Big)\Big(\frac{\nbt(a)\nbt(n-a)}{\nbt(n-1)}\Big).
$$
Together withTheorem~\ref{thm:cor:unroot:PDA}, this shows that the theorem holds for the case $A\subseteq B$, as required. 
}
\hfill $\square$
\end{proof}

\section{Discussion and concluding remarks}
Clade sizes are an important genealogical feature in the study of phylogenetic and  population genetics. In this paper we present a comparison study between the clade probabilities under the YHK and PDA models,  two null models which are commonly used in evolutionary biology.

Our first main result reveals a common feature, that is, the  clade probability sequences are log-convex under both  models. This implies that compared  with `mid-sized' clades, very `large' clades and very `small' clades are more likely to occur under these two models, and hence provides a theoretical explanation for the empirical result on the PDA model observed by~\citet{PR05}. 
One implication of this  result is that in Bayesian analysis where the two null models are used as prior distribution, the distribution on clades is not uninformative as bias is given to those whose sizes are extreme. Therefore, further considerations or adjustment, such as introducing a Bayes factor to account for the bias on prior clade probabilities, is important  to interpret posterior Bayesian clade supports.  

The second result reveals a `phase transition' type feature when comparing   the sequences of clade probabilities under the two null models. That is, we prove that there exists a critical value $\kappa(n)$ such that the probability that a given clade with size $k$ is contained  in a random tree with $n$ leaves  generated under the YHK model is smaller than that under the PDA model for  $1<k\leqslant \kappa(n)$, and  higher for all $\kappa(n)\leqslant k <n$. 
This implies that typically the trees generated under the YHK model contains relatively more `small' clades than those under the PDA model.  

The above two results are also extended to unrooted trees by considering the probabilities of `clans',  the sets of \tw{taxa} that are all on one side of an edge in an unrooted phylogenetic tree. 
This extension is relevant because in many tree reconstruction approaches, the problem of finding the root is either ignored or left as the last step. Here 
we study the sequences formed by clan probabilities for unrooted trees  \tw{generated by the two null models, and obtain several results similar to those for rooted trees.}

Note that the two models studied here are special instances of the $\beta$-splitting model introduced by~\citet{aldous96a},  a critical branching process in which the YHK model corresponds to $\beta=0$ and the PDA model to $\beta=-1.5$. Therefore, it would be of interest to study clade and clan probabilities under this more general model. In particular, it is interesting to see whether the relationships between two models revealed in this paper also hold for general $\beta$.

\begin{acknowledgements}
We thank Prof. Kwok Pui Choi and Prof. Noah A.
Rosenberg for simulating discussions and useful suggestions.
We would also like to thank two anonymous  referees for their
helpful and constructive comments on the first version of this paper.
\end{acknowledgements}

\bibliographystyle{spbasic}      
{\footnotesize
\bibliography{pdacladebib.bib}   

\begin{thebibliography}{32}
\providecommand{\natexlab}[1]{#1}
\providecommand{\url}[1]{{#1}}
\providecommand{\urlprefix}{URL }
\expandafter\ifx\csname urlstyle\endcsname\relax
  \providecommand{\doi}[1]{DOI~\discretionary{}{}{}#1}\else
  \providecommand{\doi}{DOI~\discretionary{}{}{}\begingroup
  \urlstyle{rm}\Url}\fi
\providecommand{\eprint}[2][]{\url{#2}}

\bibitem[{Agapow and Purvis(2002)}]{agapow02a}
Agapow PM, Purvis A (2002) Power of eight tree shape statistics to detect
  nonrandom diversification: a comparison by simulation of two models of
  cladogenesis. Systematic Biology 51:866--872

\bibitem[{Aldous(1996)}]{aldous96a}
Aldous D (1996) Probability distributions on cladograms. In: Aldous D, Pemantle
  R (eds) Random Discrete Structures, The IMA Volumes in Mathematics and its
  Applications, vol~75, Springer-Verlag, pp 1--18

\bibitem[{Aldous(2001)}]{Aldous2001}
Aldous D (2001) Stochastic models and descriptive statistics for phylogenetic
  trees, from {Yule} to today. Statistical Science 16(1):23--34

\bibitem[{Blum and Francois(2005)}]{blum05a}
Blum MGB, Francois O (2005) Minimal clade size and external branch length under
  the neutral coalescent. Advances in Applied Probability 37:647--662

\bibitem[{Blum et~al(2006)Blum, Francois, and Janson}]{Blum2006}
Blum MGB, Francois O, Janson S (2006) The mean, variance and limiting
  distribution of two statistics sensitive to phylogenetic tree balance. The
  Annals of Applied Probability 16(4):2195--2214

\bibitem[{Brown(1994)}]{brown94a}
Brown JKM (1994) Probabilities of evolutionary trees. Systematic Biology
  43(1):78--91

\bibitem[{Colless(1982)}]{Colless1982}
Colless DH (1982) Review of ``{P}hyogenetics: {T}he theory and practice of
  phylogenetic systematics''. Systematic Zoology 31:100--104

\bibitem[{Felsenstein(2004)}]{felsenstein04a}
Felsenstein J (2004) Inferring Phylogenies. Sinauer Associates, Sunderland, MA

\bibitem[{Harding(1971)}]{harding71a}
Harding EF (1971) The probabilities of rooted tree-shapes generated by random
  bifurcation. Advances in Applied Probability 3(1):44--77

\bibitem[{Heard(1992)}]{heard92a}
Heard SB (1992) Patterns in tree balance among cladistic, phenetic, and
  randomly generated phylogenetic trees. Evolution 46(6):1818--1826

\bibitem[{Hudson and Coyne(2002)}]{hudson02a}
Hudson RR, Coyne JA (2002) Mathematical consequences of the genealogical
  species concept. Evolution 56(8):1557--1565

\bibitem[{Kingman(1982)}]{Kingman1982}
Kingman JFC (1982) On the genealogy of large populations. Journal of Applied
  Probability 19:27--43

\bibitem[{Li et~al(2000)Li, Pearl, and Doss}]{Li2000}
Li S, Pearl DK, Doss H (2000) Phylogenetic tree construction using {Markov
  Chain Monte Carlo}. Journal of the American Statistical Association
  95(450):493--508

\bibitem[{Liu and Wang(2007)}]{LW}
Liu LL, Wang Y (2007) On the log-convexity of combinatorial sequences. Advances
  in Applied Mathematics 39:453--476

\bibitem[{McKenzie and Steel(2000)}]{McKenzie2000}
McKenzie A, Steel MA (2000) Distributions of cherries for two models of trees.
  Mathematical Biosciences 164:81--92

\bibitem[{Mooers and Heard(1997)}]{mooers97a}
Mooers AO, Heard SB (1997) Evolutionary process from phylogenetic tree shape.
  Quarterly Review of Biology 72:31--54

\bibitem[{Mooers and Heard(2002)}]{mooers02a}
Mooers AO, Heard SB (2002) Using tree shape. Systematic Biology 51:833--834

\bibitem[{Nordborg(1998)}]{nordborg98a}
Nordborg M (1998) On the probability of {N}eanderthal ancestry. American
  Journal of Human Genetics 63:1237--1240

\bibitem[{Nordborg(2001)}]{nordborg01a}
Nordborg M (2001) Coalescent theory. In: Balding DJ, Bishop M, Cannings C (eds)
  Handbook of Statistical Genetics, Wiley, Chichester, UK, chap~7, pp 179--212

\bibitem[{Pickett and Randle(2005)}]{PR05}
Pickett KM, Randle CP (2005) Strange bayes indeed: uniform topological prior
  imply non-uniform clade priors. Molecular Phylogenetics and Evolution
  34:203--211

\bibitem[{Pinelis(2003)}]{Pinelis2003}
Pinelis I (2003) Evolutionary models of phylogenetic trees. Proceedings of the
  Royal Society of London Series B: Biological Sciences 270:1425--1431

\bibitem[{Rannala and Yang(1996)}]{Rannala1996}
Rannala B, Yang Z (1996) Probability distribution of molecular evolutionary
  trees: a new method of phylogenetic inference. Journal of Molecular Evolution
  43:304--311

\bibitem[{Rogers(1996)}]{rogers96a}
Rogers JS (1996) Central moments and probability distributions of three
  measures of phylogenetic tree imbalance. Systematic Biology 45:99--110

\bibitem[{Rosenberg(2003)}]{rosenberg03a}
Rosenberg NA (2003) The shapes of neutral gene genealogies in two species:
  probabilities of monophyly, paraphyly and polyphyly in a coalescent model.
  Evolution 57(7):1465--1477

\bibitem[{Rosenberg(2006)}]{rosenberg06a}
Rosenberg NA (2006) The mean and variance of the numbers of r-pronged nodes and
  r-caterpillars in {Yule}-generated genealogical trees. Annals of
  Combinatorics 10:129--146

\bibitem[{Rosenberg(2007)}]{rosenberg07a}
Rosenberg NA (2007) Statistical tests for taxonomic distinctiveness from
  observations of monophyly. Evolution 61(2):317--323

\bibitem[{Sackin(1972)}]{Sackin1972}
Sackin MJ (1972) ``{Good}" and ``{bad}" phenograms. Systematic Zoology
  21(2):225--226

\bibitem[{Semple and Steel(2003)}]{semple03a}
Semple C, Steel MA (2003) Phylogenetics. Oxford University Press, Oxford, UK

\bibitem[{Steel and Pickett(2006)}]{PS06}
Steel M, Pickett KM (2006) On the impossibility of uniform priors on clades.
  Molecular Phylogenetics and Evolution 39:585--586

\bibitem[{Steel(2012)}]{Steel2012}
Steel MA (2012) Root location in random trees: A polarity property of all
  sampling consistent phylogenetic models except one. Molecular Phylogenetics
  and Evolution 65(1):345 -- 348

\bibitem[{Yule(1925)}]{yule25a}
Yule GU (1925) A mathematical theory of evolution. based on the conclusions of
  {Dr. J.C. Willis, F.R.S}. In: Philosophical Transactions of the Royal Society
  of London. Series B, Containing Papers of a Biological Character, vol 213,
  The Royal Society, pp 21--87

\bibitem[{Zhu et~al(2011)Zhu, Degnan, and Steel}]{zhu11a}
Zhu S, Degnan JH, Steel MA (2011) Clades, clans and reciprocal monophyly under
  neutral evolutionary models. Theoretical Population Biology 79:220--227

\end{thebibliography}
}

\end{document}